\documentclass[sigconf]{acmart}

\usepackage{booktabs} 

\usepackage{graphicx}
\usepackage{subfigure}

\newcommand{\be}{\begin{equation}}
\newcommand{\ee}{\end{equation}}
\newcommand{\bb}{\mathbf{b}}
\newcommand{\hbb}{\mathbf{\hat{b}}}
\newcommand{\vvv}{\mathbf{v}}
\newcommand{\dd}{\mathbf{d}}

\newcommand{\uu}{\mathbf{u}}
\newcommand{\zz}{\mathbf{z}}
\newcommand{\aaa}{\mathbf{a}}
\newcommand{\rr}{\mathbf{r}}

\setcopyright{none}
\renewcommand\footnotetextcopyrightpermission[1]{} %

\begin{document}

\title{Pre- and Post-Auction Discounts in First-Price Auctions}
\author{\huge
\begin{tabular}{cc}
\\
\bfseries{Miguel Alcobendas} & \bfseries{Eric Bax} \\
\Large Yahoo Research & \Large Yahoo Research \\
\large \texttt{lisbona@yahooinc.com} & \large \texttt{ebax@yahooinc.com}
\end{tabular}}

\begin{abstract}
One method to offer some bidders a discount in a first-price auction is to augment their bids when selecting a winner but only charge them their original bids should they win. Another method is to use their original bids to select a winner, then charge them a discounted price that is lower than their bid should they win. We show that the two methods have equivalent auction outcomes, for equal additive discounts and for multiplicative ones with appropriate adjustments to discount amounts. As a result, they have corresponding equilibria when equilibria exist. We also show that with the same level of multiplicative adjustments, bidders with discounts should prefer an augmented bid to a discounted price. Then we estimate optimal bid functions for valuation distributions based on data from online advertising auctions, and show how different discount levels affect auction outcomes for those bid functions.
\newline
\newline
\newline
\newline
\newline
\end{abstract}


\settopmatter{printfolios=true}
\maketitle
\pagestyle{plain}

\section{Introduction} \label{sec_introduction}
First-price auctions are now the norm in online advertising exchanges \cite{bigler2019,cox2019,akbarpour20,despotakis21}: bidders enter bids, the bidder with the highest bid wins the right to display an advertisement, and they pay their bid. An exchange or a seller may offer some bidders discounts for a variety of reasons, including funding new bidders' learning on the exchange, increasing demand and competition in auctions, or in fulfillment of contract conditions. In online advertising, as in other types of media like TV or radio, it is common for publishers to offer auction discounts to ad agencies or advertisers with large budgets, hoping that the discount attracts more ad spending from the buyer. Outside of online advertising, discounting in auctions has been a way to increase participation from underrepresented or disadvantaged groups of bidders \cite{marion07,kras11}. 

A discount may be applied either pre-auction or post-auction. With a pre-auction discount, which we refer to as an augmented bid, the bid is increased for the purpose of determining the auction winner, but not for determining the winner's price. With a post-auction discount, the winner is selected based on the original bids, but if a bidder with a discount wins, then their winning price is a discounted version of their bid. We will refer to a pre-auction discount as bid augmentation and a post-auction discount as price reduction. 

An additive discount is a fixed amount added to a bid when selecting a winner or subtracted from a bid when charging the bidder should they win. A multiplicative discount is a fraction of the bid that is either added to the bid or subtracted from the price -- for example, a 5\% bid augmentation or discount. In this paper, we consider both additive and multiplicative discounts. Additive discounts are simpler to analyze; multiplicative discounts occur more often in practice. 

We analyze equilibria for first-price auctions with discounts, comparing equilibria under pre- and post-auction discounts. There is a rich literature on equilibria in first-price auctions. Under continuity and independence constraints on bidder valuations, it is possible to show that first-price auctions have unique equilibria and to characterize the equilibria \cite{milgrom82,lebrun99,mcfadden03,levin04,reny04,milgrom05} in terms of conditional second values (CSV), with each bidder bidding the expectation of the runner-up valuation conditioned on their own valuation being the highest. Some challenges for more general methods include dependence among valuations (common values), bidders having different valuation distributions (asymmetry), and the fact that there are simple examples of non-continuous valuation distributions for which equilibria do not exist \cite{maskin00,wang20}. More recently, there have been advances in computing estimated equilibria via an alternative tie-breaking method and discrete approximations \cite{maskin00,maskin03}, indirect methods based on payoff ratios \cite{kirkegaard09}, dynamical systems methods \cite{fibich12}, polynomial basis functions \cite{hubbard12,hubbard13}, and fictitious play \cite{heymann21}.  There is also a method to characterize Bayesian Nash equilibrium bidding strategies for bidders with some discrete valuation distributions \cite{wang20}. 

In Section \ref{additive_sec}, we show that the same additive discounts applied pre- and post-auction lead to the same incentives and outcomes for bidders and the seller, implying the same equilibria. Section \ref{mult_sec} presents an equivalence between (non-equal) multiplicative discount levels applied pre- and post-auction, shows how equal pre- and post-auction discounts impact bidders and the seller, and considers how discounts can produce symmetry in a specific case. Section \ref{emp_sec} presents results based on data from online advertising auctions, explaining how to estimate bid functions under different levels of discounting and examining auction outcomes based on those bid functions. Section \ref{conc_sec} concludes with some potential directions for future work.

\section{Additive Discounts and Matching Equilibria} \label{additive_sec}
In this section, we show that for additive discounts, bid augmentations and price reductions at the same levels lead to equivalent equilibria, in the sense that they have the same outcomes for all participants. Let vectors $\vvv$ be the bidders' valuations, $\bb$ be their bids, and $\aaa$ their bid augmentations (the amounts added to the bids solely for the purpose of selecting the winner). Let $v_i$, $b_i$, and $a_i$ be bidder $i$'s valuation, bid, and bid augmentation. (If bidder $i$ has no discount, then $a_i = 0$.) With bid augmentations, bidder $i$ has utility function
\be
u^+_i(\vvv, \bb, \aaa) = p_i(\bb + \aaa) (v_i - b_i)
\ee
where $p_i()$ is the probability that bidder $i$ wins, which is one if bidder $i$ has the greatest augmented bid:
\be
\forall j \not= i: \hbox{ } b_i + a_i > b_j + a_j,
\ee
zero if some other bidder has a greater augmented bid, and $\frac{1}{k}$ if $k$ bidders, including bidder $i$, tie for maximum augmented bid. 

With price reductions, let $\vvv$ and $\bb$ be valuations and bids as before, and let $\rr$ be the vector of price reductions, with $r_i$ the price reduction for bidder $i$. Then bidder $i$ has utility function:
\be
z^+_i(\vvv, \bb, \rr) = p_i(\bb) (v_i - b_i + r_i).
\ee

\begin{theorem}
For additive discounts, first-price auctions with bid augmentation and those with price reductions have corresponding Nash equilibria (if either has an equilibrium), with each bidder having that same utility under both forms of discounting, if the bid augmentations are equal to the price reductions. 
\end{theorem}

\begin{proof}
Let $\uu^+$ and $\zz^+$ be the vectors of utility functions for bid augmentation and price reduction, respectively. Note that
\be
\zz^+(\vvv, \bb, \rr) = \uu^+(\vvv, \bb - \rr, \rr)
\ee
since, for each bidder $i$
\be
u^+_i(\vvv, \bb - \rr, \rr) = p_i(\bb - \rr + \rr) (v_i - (b_i - r_i)) = p_i(\bb) (v_i - b_i + r_i) = z^+_i(\vvv, \bb, \rr).
\ee

Prove by contradiction. Let $\bb^*$ be any Nash equilibrium bids for price reductions. Now suppose $\bb^* - \rr$ are not Nash equilibrium bids for bid augmentation.  Then some bidder $i$ can change their bid and increase their own utility.  Let $\bb^+ - \rr$ be the boosting bids after bidder $i$ changes their bid to one that increases their utility. Then
\be
z^+_i(\vvv, \bb^*, \rr) = u^+_i(\vvv, \bb^* - \rr, \rr) < u^+_i(\vvv, \bb^+ - \rr, \rr) = z^+_i(\vvv, \bb^+, \dd),
\ee
but $z^+_i(\vvv, \bb^*, \rr) < z^+_i(\vvv, \bb^+, \rr)$ implies that $\bb^*$ is not a Nash equilibrium. So $\bb^* - \rr$ are Nash equilibrium bids for bid augmentation. Since
\be
\zz^+(\vvv, \bb^*, \rr) = \uu^+(\bb^* - \rr, \rr),
\ee
each bidder's utility is the same for bid augmentation as for price reduction, for these corresponding equilibria.

Similar reasoning shows that each Nash equilibrium for bid augmentation corresponds to a Nash equilibrium for price reduction, with augmentations equal to reductions, and with the same utilities for each bidder.
\end{proof}

\section{Multiplicative Discounts and Corresponding Equilibria} \label{mult_sec}
In this section, we explore how multiplicative discounts affect auctions. We show that there are price reductions that produce auctions with equilibria that have the same utilities for all participants as those for bid-augmentation auctions. Next, we compare price-reduction and bid-augmentation auctions that use the same discount rates. Then, we show how discounts can be set to remove asymmetry among bidders with some uniform valuation distributions. 

Under multiplicative discounts, let $a_i$ be the bid augmentation rate for bidder $i$, so that their bid is treated as $\hat{b}_i = b_i (1 + a_i)$ when determining the winner, but they pay only their bid $b_i$ should they win. Define vector $\hbb$ to be the vector of augmented bids $\hat{b}_i$. Then bidder $i$ has utility function
\be
u_i(\vvv, \bb, \aaa) = p_i(\hbb) (v_i - b_i).
\ee
Note that
\be
b_i = \frac{\hat{b}_i}{1 + a_i} = \frac{\hat{b}_i + \hat{b}_i a_i - \hat{b}_i a_i}{1 + a_i} = \hat{b}_i - \hat{b}_i \left(\frac{a_i}{1 + a_i}\right).
\ee
So bidder $i$'s utility function 
\be
u_i(\vvv, \bb, \aaa) = p_i(\hbb) (v_i - b_i) = p_i(\hbb) \left[\left(v_i + \hat{b}_i \left(\frac{a_i}{1 + a_i}\right)\right) - \hat{b}_i\right].
\ee
This is equal to a utility function without an augmented bid ($a_i = 0$), except with bidder $i$'s valuation not $v_i$ but rather a ``virtual validation'' of $\hat{v}_i = v_i + \hat{b}_i \left(\frac{a_i}{1 + a_i}\right)$.

For price reductions, let $r_i$ be the reduction rate for bidder $i$, with their original bid $b_i$ determining whether they win, but their price being $b_i (1 - r_i)$ should they win. Then bidder $i$ has utility function 
\be
z_i(\vvv, \bb, \rr) = p_i(\bb) (v_i - b_i (1 - r_i)).
\ee
Note that this is 
\be
 = p_i(\bb) \left[(v_i + b_i r_i) - b_i\right].
\ee
This is equal to a utility function without a price reduction ($r_i = 0$), but with bidder $i$'s valuation a ``virtual valuation" $\tilde{v}_i = v_i + b_i r_i$ rather than $v_i$.

\subsection{Corresponding Equilibria}
Setting price reduction levels $r_i$ to make virtual valuations the same for a price-reduction auction as those for a bid-augmentation auction produces auctions with corresponding equilibria:

\begin{theorem} \label{mult_equiv}
Setting $r_i = \frac{a_i}{1 + a_i}$ for each bidder $i$ creates corresponding equilibria (if any) for a first-price auction with bid adjustments as for one with price reductions, assuming each bidder $i$ has the same distribution of valuations for both auctions. The corresponding equilibria have the same expected utilities for all participants. 
\end{theorem}

\begin{proof}
With bid augmentation, the utility function for bidder $i$ is
\be
u_i(\vvv, \bb, \aaa) = p_i(\hbb) (v_i - b_i),
\ee
with $\hbb$ the vector of augmented bids $\hat{b}_i = b_i (1 + a_i)$. For a price-reduction auction, suppose each bidder bids $\hat{b}_i$. Then bidder $i$ has utility
\be
z_i(\vvv, \hbb, \rr) = p_i(\hbb) (v_i - \hat{b}_i (1 - r_i)) = p_i(\hbb) (v_i - b_i (1 + a_i) (1 - r_i)). \label{eq_ut}
\ee
But this is equal to $u_i(\vvv, \bb, \aaa)$ if 
\be
b_i = b_i (1 + a_i) (1 - r_i).
\ee
Setting $1 = (1 + a_i) (1 - r_i)$ and solving for $r_i$:
\be
r_i = \frac{a_i}{1 + a_i}.
\ee
So those price reductions $r_i$, with bidders augmenting their own bids to $\hat{b}_i$, produce the same utility function for each bidder under price reductions as they would have with bids $b_i$ under bid augmentations $a_i$, with the same bids $\hbb$ determining allocations. So the auctions are equivalent, in the sense that bidders have an equivalent set of choices with the same impact on each other's outcomes and their own, and hence they have corresponding equilibria with the same expected utility for each bidder. 

We can prove that by contradiction. Suppose $\bb^*()$ is a vector of equilibrium bid functions: $[b_1^*(), \ldots, b_n^*()]$, for bid augmentation with augmentation $a_i$ for each bidder $i$. Let $r_i = \frac{a_i}{1 + a_i}$ for each bidder $i$. Then, by Equation \ref{eq_ut}, all bidders receive the same expected utilities for price reduction as for bid augmentation if they use the vector of bid functions $\hbb^*() = [(1 + a_1) b_1^*(), \ldots, (1 + a_n) b_n^*()]$. Now suppose these are not equilibrium bid functions. Then some bidder can improve their expected utility for the price-reduction auction by using a different bid function. Without loss of generality, let it be bidder 1, and let the bid function that improves their expected utility be $\tilde{b}_1() \not= (1 + a_1) b_1^*()$. Then, they can use bid function $\tilde{b}_1() / (1 + a_1)$ to improve their expected utility for the bid-augmentation auction. But that contradicts $\bb^*()$ being a vector of equilibrium bid functions.
\end{proof}

\subsection{Bid Augmentation vs. Price Reduction with Equal Rates} \label{eq_sec}
Based on Theorem \ref{mult_equiv}, changing from multiplicative bid augmentations $a_i$ to price reductions at the same level, $r_i = a_i$, has the same effect on auction equilibria and outcomes as lowering the bid augmentations to
\be
\frac{a_i}{1 + a_i} = \frac{a_i + a_i^2 -a_i^2}{1 + a_i} = \frac{a_i(1 + a_i) - a_i^2}{1 + a_i} = a_i - \frac{a_i^2}{1 + a_i},
\ee
which is lowering them by $\frac{a_i^2}{1 + a_i}$. Since it is equivalent to lowering the bid augmentation rates, changing from bid augmentations to equal price reduction rates is a disadvantage for bidders with discounts, and a greater disadvantage for bidders with greater discounts. For small bid augmentation rates, the disadvantage is mild. For example, with a bid augmentation rate of $a_i = 10\%$, changing to a price reduction of $r_i = 10\%$ is equivalent to lowering the bid augmentation rate by $\frac{1}{110}$, or about 1\%.  For a 5\% discount, the discount rate loses about a quarter of a percent, while a 25\% discount becomes equivalent to a 20\% discount instead.

These changes are small relative to the discount amounts, but they may still be large relative to the bidder's expected profit, since the difference between a bidders price and their valuation may be significantly smaller than their price. Hence, changing from one type of discount to the other may affect bidder's expected surplus (profit) significantly, even to the point of determining whether participation produces a positive profit.

\subsection{Optimal Bidding for Uniform Value Distributions with Compensating Discounts}
Assume $n$ bidders each have independent value distributions $U[0,1]$. Without discounts, Bayes-Nash equilibrium bids for this case are \cite{vickrey61,myerson81,riley81,levin04,milgrom05}
\be
b_i = \frac{n-1}{n} v_i,
\ee
where $v_i$ is the value for bidder $i$, and $b_i$ is their optimal (equilibrium) bid. Since the conditions for the revenue-equivalence theorem \cite{vickrey61,myerson81,riley81,maskin00asy,milgrom05} hold for this case -- including having independent and identical (also called symmetric) value distributions -- this result can be proven by computing conditional second values (CSV) \cite{milgrom05} (Thm. 4.6): conditioned on $v_i$ being the top value, the other $n - 1$ bidders' values are independently distributed over $U[0, v_i]$, so their top value as a fraction of $v_i$ has the same distribution as the top value among $n - 1$ independent draws from $U[0, 1]$. That has a $\mathrm{Beta}(n - 1, 1)$ distribution \cite{david03}, which has mean $\frac{n - 1}{n}$.

Now consider the impact of price reductions for independent uniform distributions of values. For each bidder $i$, a price reduction of $r_i$ boosts their equilibrium bid by a factor of $\frac{1}{1 - r_i}$, for the case of independent uniform value distributions $U[0, 1 - r_i]$ rather than $U[0,1]$:

\begin{theorem} \label{uniform_thm}
If each bidder $i$ has price reduction $r_i$ and $n$ bidders each have independent value distributions $U[0,1 - r_i]$, then 
\be
b_i = \left(\frac{n-1}{n}\right) \left(\frac{1}{1 - r_i}\right) v_i
\ee
are Bayes-Nash equilibrium bids.
\end{theorem}

\begin{proof}
If bidder $i$ bids $b_i$ and bidder $j$ bids $b_j = \left(\frac{n-1}{n}\right) \left(\frac{1}{1 - r_j}\right) v_j$ then $b_j \leq b_i$ iff
\be
\left(\frac{n-1}{n}\right) \left(\frac{1}{1 - r_j}\right) v_j \leq b_i.
\ee
Equivalently,
\be
 \left(\frac{1}{1 - r_j}\right) v_j \leq \left(\frac{n}{n - 1}\right) b_i.
\ee
Since each $ \left(\frac{1}{1 - r_j}\right) v_j \sim U[0, 1]$ and $\left(\frac{n}{n - 1}\right) b_i \leq 1$, the probability that $b_j \leq b_i$ is 
\be
\left(\frac{n}{n - 1}\right) b_i.
\ee
Since the product is the probability that bidder $i$ has a top bid, and $v_i - b_i (1 - r_i)$ is bidder $i$'s surplus if they win, bidder $i$ has expected utility:
\be
\mathrm{E} z_i = \prod_{j \not= i} \left[\left(\frac{n}{n - 1}\right) b_i\right] \left[v_i - b_i (1 - r_i)\right].
\ee
Simplifying the product,
\be
\mathrm{E} z_i = \left(\frac{n}{n - 1}\right)^{n - 1} b_i^{n - 1} \left[v_i - b_i (1 - r_i)\right],
\ee

To show that $b_i = \left(\frac{n - 1}{n}\right) \left(\frac{1}{1 - r_i}\right) v_i$ is an optimal bid, take the derivative:
\be
\frac{\partial}{\partial b_i} \mathrm{E} z_i = \left(\frac{n}{n - 1}\right)^{n - 1}  \left[(n - 1) b_i^{n - 2} \left[v_i - b_i (1 - r_i)\right] - b_i^{n - 1} (1 - r_i)\right].
\ee
Then set it to zero:
\be
(n - 1) b_i^{n - 2} \left[v_i - b_i (1 - r_i)\right] = b_i^{n - 1} (1 - r_i).
\ee
Equivalently,
\be
(n - 1) \left[v_i - b_i (1 - r_i)\right] = b_i (1 - r_i),
\ee
and
\be
(n - 1) v_i = n b_i (1 - r_i).
\ee
Solve for $b_i$:
\be
b_i = \left(\frac{n-1}{n}\right) \left(\frac{1}{1 - r_i}\right) v_i.
\ee
\end{proof}

Theorem \ref{uniform_thm} is based on uniform value distributions, but over domains $U[0,1 - r_i]$ rather than $U[0, 1]$ -- distributions with values reduced to match the price discounts. In equilibrium, bidders boost their bids (relative to the symmetric $U[0, 1]]$-distribution case) enough to compensate for their price discounts, making the winner's price $\frac{n - 1}{n}$ times their post-discount values. So discounts even the playing field, giving the asymmetric $U[0, 1 - r_i]$ auction dynamics similar to a symmetric $U[0, 1]$ auction.

\section{Bid Functions and Auctions Under Realistic Valuations} \label{emp_sec}
 In online advertising, as in other type of media like TV or radio, ad agencies or advertisers with large budgets may benefit from discounts offered by publishers designed to attract ad spending. This section explores auction outcomes under discounting for bidders with realistic valuation distributions using online advertising data from Yahoo! ad exchange. Amounts (revenue, valuations, bids, costs, and surpluses) are all multiplied by the same secret constant to obscure actual amounts but preserve relative relationships. Original bids are in US dollars per 1000 impressions -- cost per mille (CPM). 

\subsection{High-Level Process}
We consider an asymmetric first price auction model with two types of buyers: a bidder that benefits from discounts, and the rest. We use the following process to recover the valuation distribution of bidders from bid data, estimate the bidding functions, and estimate expected auction outcomes when one of the buyers benefits from a discount:

\begin{enumerate}
\item Collect bid data for an ad position on the Yahoo! ad platform that has a 5\% price reduction discount for some bidders. The data is for about 200\,000 auctions, drawn from a day in February 2023. The average number of participants in the auction data is around five. For our analysis, we assume that only one of the five bidders receives a discount.
\item Assuming that valuation distributions follow a log-normal distribution and valuation distributions are i.i.d. (no affiliation between bidders), apply an approach by Guerre, Perrigne, and Vuong \cite{guerre00,ma19} to the historical bids, then use maximum likelihood to estimate the shape and scale parameters for two valuation distributions: one for discounted bidders and the other for other bidders. The resulting valuation distributions are similar to each other, having shape and scale parameters that each differ by less than 3\%. 
\item For a discounted bidder with the discounted bidders' valuation distribution and four other bidders with the other valuation distribution, numerically estimate the bid functions for the discounted bidder and for the other bidders. To do this, we apply Euler's method to a pair of differential equations to compute estimated inverse bid functions, based on ideas from Maskin and Riley \cite{maskin00asy} and differential equations from \cite{gentry18}, modified to reflect a price reduction. Then we invert to produce estimated bid functions. (For details, see Subsection \ref{subsec_est_bf}.)
\item Based on those estimated bid functions, compute auction outcomes for auctions with a discounted bidder and four other bidders, under different price reduction rates. The outcomes include seller revenue and bidder surplus for the discounted and other bidders. We present those results in Subsection \ref{subsec_results}.
\end{enumerate}

\subsection{Estimating Bid Functions} \label{subsec_est_bf}
To estimate bid functions for a single discounted bidder and four undiscounted bidders, we use differential equations similar to those from \cite{gentry18}. We modify them for a bid reduction, and, instead of fitting a polynomial, we compute based on Euler's method. Following Maskin and Riley \cite{maskin00asy}, we first estimate inverse bid functions -- valuations as functions of bids that are optimal for the valuations -- then invert to produce bid functions. Also based on ideas from Maskin and Riley \cite{maskin00asy}, we ``guess and check" to identify a maximum bid that implies a feasible boundary condition for the system of differential equations. We assume the undiscounted bidders have identical valuation distributions and, by symmetry, equal bid functions, so we only explicitly solve for one of them.

Let bidder 1 be a discounted bidder, with a price reduction that is $r$ times their bid. Let bidders 2 to 5 be undiscounted. Let $F_1$ and $f_1$ be the cdf and pdf, respectively, of the valuation distribution for bidder 1. Let $F_2$ and $f_2$ be those functions for the valuation distributions for bidders 2 to 5. Let $v_1$ and $v_2$ be the valuations for bidders 1 and 2 to 5, respectively, for which $b$ is the optimal bid. Use $F_i$ and $f_i$ as shorthand for $F_i(v_i)$ and $f_i(v_i)$. Use an apostrophe to denote a derivative with respect to $b$, and use $n$ for the number of bidders without discounts. (For us, $n = 4$.)

The utility function for bidder 1 is
\be
u_1 = [v_1 - (1 - r) b] F_2^n,
\ee
since all $n$ undiscounted bidders must have valuations below $v_2$ to bid less than $b$. (Assume continuous valuation distributions to ignore ties.) The utility function for bidder 2 is
\be
u_2 = (v_2 - b) F_1 F_2^{n - 1},
\ee
since the discounted bidder and the other $n - 1$ undiscounted bidders must all bid less than bidder 2 for bidder 2 to win. 

For $b$ to be optimal for $v_1$ and $v_2$, the derivative of utility functions with respect to $b$ must be zero. Since
\be
u_1' = [v_1 - (1 - r) b] n F_2^{n - 1} f_2 v_2' - (1 - r) F_2^n
\ee
and
\be
u_2' = (v_2 - b) \left[f_1 v_1' F_2^{n - 1} + (n - 1) F_1 F_2^{n - 2} f_2 v_2'\right] - F_1 F_2^{n - 1},
\ee
we can set each RHS to zero, solve the first for $v_2'$:
\be
v_2' = \frac{(1 - r) F_2}{n f_2 [v_1 - (1 - r) b]}, \label{v2_eqn}
\ee
and solve the second for $v_1'$:
\be
v_1' = \frac{F_1}{f_1}\left(\frac{1}{v_2 - b} - \frac{(n - 1) f_2 v_2'}{F_2}\right). \label{v1_eqn}
\ee

Let $v_1^*$ and $v_2^*$ be the maximum values in the support of the valuation distributions for bidders 1 and 2, respectively. Given a maximum bid $b^*$ (as yet unknown), for which $v_1^*$ and $v_2^*$ are optimal, we know $F_1(v_1^*) = F_2(v_2^*) = 1$ and we know $f_1(v_1^*)$ and $f_2(v_2^*)$. Our method substitutes those values for $F_1$, $F_2$, $f_1$, and $f_2$ and $v_1^*$ and $v_2^*$ for $v_1$ and $v_2$, with $b$ set to $b^*$ in Equations \ref{v2_eqn} and \ref{v1_eqn}. We then set the inverse bid function for $v_1$ at $b^* - \Delta_b$ (for a small $\Delta_b$) to $v_1^* - v_1' \Delta_b$ and the inverse bid function for $v_2$ at $b^* - \Delta_b$ to $v_2^* - v_2' \Delta_b$. In other words, we use Euler's method, computing the derivatives $v_1'$ and $v_2'$ at $b^*$ and using a linear estimation for values of $v_1$ and $v_2$ of the inverse bid functions at $b^* - \Delta_b$. 

Then we repeat Euler's method, using the values at $b = b^* - \Delta_b$ to estimate those at $b = b^* - 2 \Delta_b$, and so on, until we reach the minimum bid of zero. We use $\Delta_b = \frac{b^*}{10000}$, yielding ten thousand and one pairs $(b, v_i)$ in each estimated inverse bid function.  To identify $b^*$, we use a binary search over candidate values. The greatest value that results in a solution with all successive values decreasing as bids decrease in the inverse bid functions is $b^*$. 

After inversion, we find that the bid functions are very close to optimal. To test this, we identified the best responses for bidder 1 and bidder 2 to the other bidders' bids across the 10001 bids in each bid function. On average, the difference between the best-response bid and the bid function bid was less than 0.03\% of the range of bids for bidder 1 and 0.002\% for bidder 2. 

Figure \ref{bidfuns} shows the estimated optimal bid functions that result from Euler's method. Figure \ref{bidfuns_detail} shows more detail for the higher valuations. For no discount ($r = 0.00$), the discounted bidder has almost the same bid function as the other four bidders, because the valuation distribution for the discounted bidder is nearly the same as that for the other bidders. For positive discounts, the bidder receiving a discount (with bid function shown by a dotted line) bids higher than the other bidders at the same valuations, since the discounted bidder receives a price discount that is $r$ times their bid. This gap increases with the discount. Also, as the discount increases, the discounted bidder exerts more competitive pressure on the other bidders, forcing them to increase their bids, which raises the bid functions and the maximum bid. 

\begin{figure}
\includegraphics[width=3.5in]{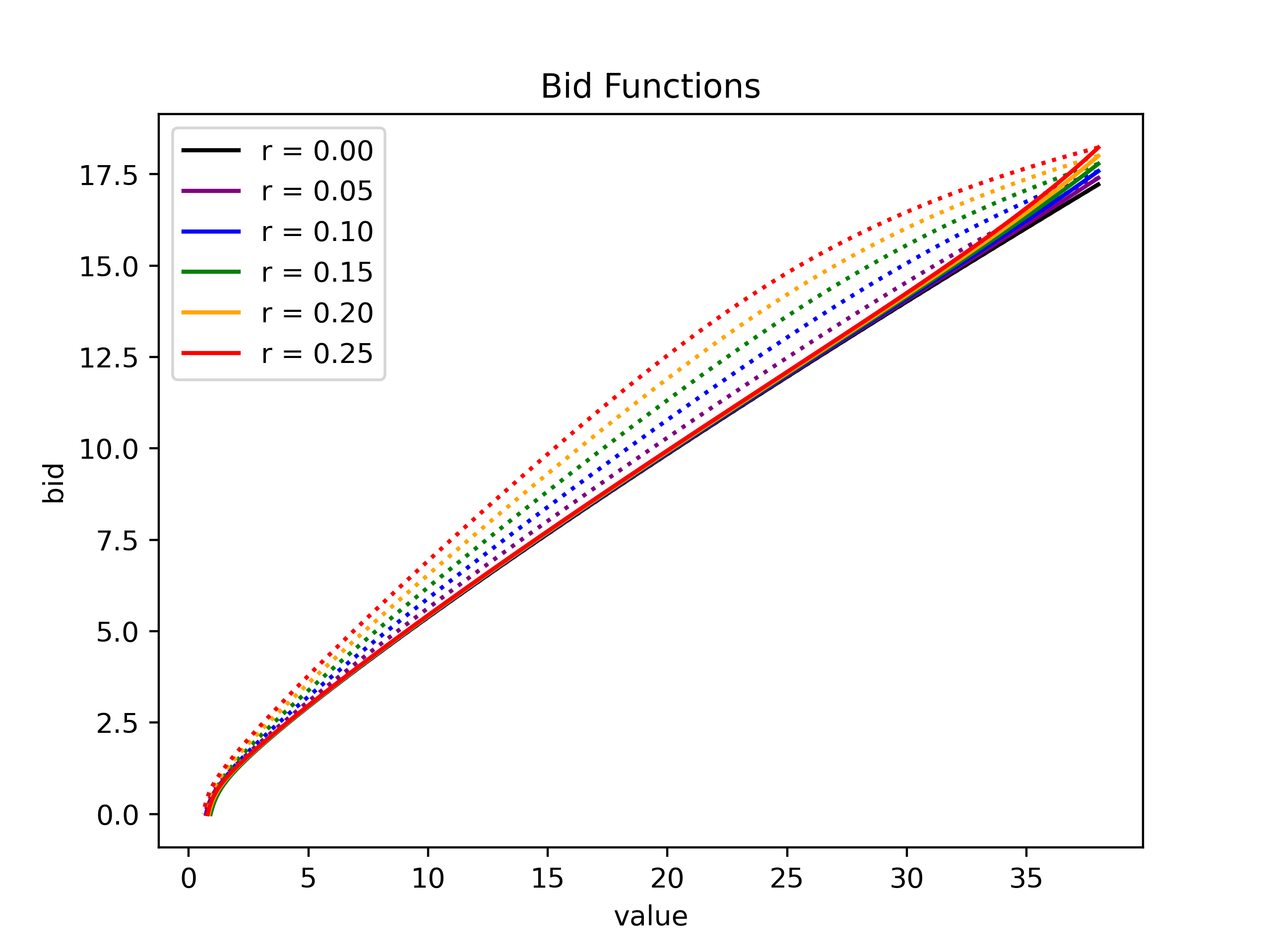} 
\caption{Estimated optimal bids as a function of valuations for a bidder with a price reduction rate $r$ (dashed line) and one of 4 undiscounted bidders (solid line) in auctions with the discounted bidder.} \label{bidfuns}
\end{figure}

\begin{figure}
\includegraphics[width=3.5in]{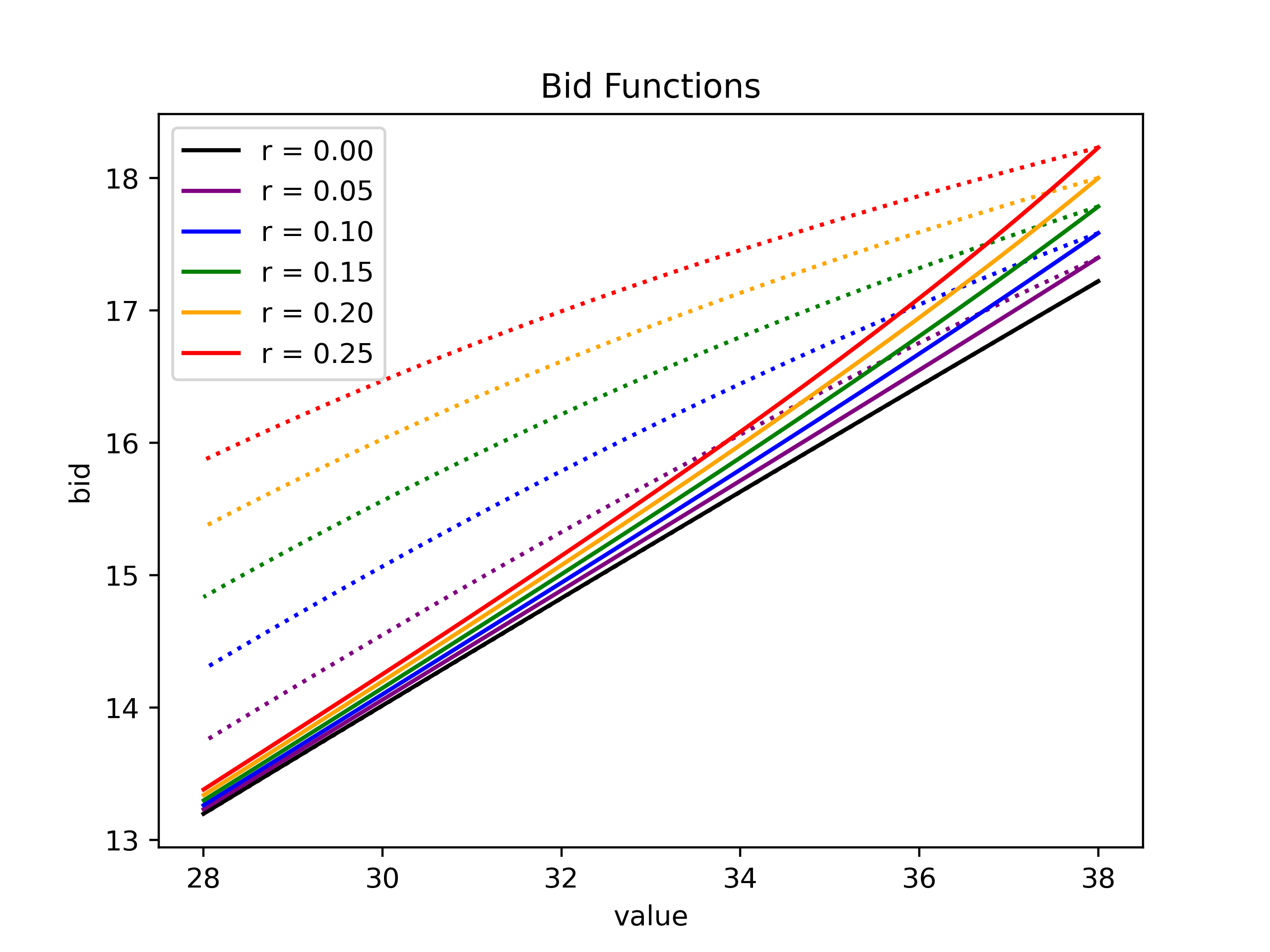} 
\caption{Estimated optimal bids as a function of valuations for a bidder with a price reduction rate $r$ (dashed line) and one of 4 undiscounted bidders (solid line) in auctions with the discounted bidder, showing detail for higher values.} \label{bidfuns_detail}
\end{figure}

\subsection{Auction Outcomes} \label{subsec_results}
\begin{table}
\begin{tabular}{|c|c|c|c|c|c|c|c|c|}
\hline
\multicolumn{9}{|c|}{Auction Outcome Statistics} \\
\hline
\multicolumn{1}{|c|}{} & \multicolumn{1}{c|}{} & \multicolumn{1}{c|}{} & \multicolumn{2}{c|}{Win Rate} & \multicolumn{2}{c|}{E Surplus} & \multicolumn{2}{c|}{E Cost}\\
\cline{4-9}
\multicolumn{1}{|c|}{r} & E Rev. & Eff. & disc. & other & disc. & other & disc. & other \\
\hline
0.00&9.320&1.00&0.200&0.200&1.977&1.975&1.865&1.864\\
0.05&9.319&0.99&0.209&0.198&2.067&1.951&1.904&1.854\\
0.10&9.299&0.98&0.218&0.196&2.165&1.930&1.936&1.841\\
0.15&9.280&0.97&0.227&0.193&2.269&1.904&1.965&1.829\\
0.20&9.255&0.96&0.237&0.191&2.379&1.877&1.989&1.817\\
0.25&9.219&0.95&0.248&0.188&2.498&1.848&2.007&1.803\\
\hline
\end{tabular}
\caption{Expected revenue for the seller ($E\;Rev.$), efficiency ($Eff.$), and win rate, expected buyer surplus ($E\;Surplus$), and expected buyer cost ($E\;Cost$) for discounted ($disc.$) and other bidders, by price reduction rate $r$. To condition surplus and cost on winning, divide by win rate.}
\label{stats_table}
\end{table}

Table \ref{stats_table} shows auction statistics based on the estimated bid functions for the discounted bidder and the four other bidders. The values are from full integration over the bid functions (not random samples from them). Because bids are discretized in the estimated bid functions, ties become possible, and they are considered in the computation of the statistics. (With 10\,001 different bid values, ties have a very small impact on the results.) 

As discount rates $r$ increase, there are mild declines in expected seller revenue ($E\;Rev.$), since a portion of the price is refunded when the discounted bidder wins. That is partially offset by the rise in bids shown in Figures \ref{bidfuns} and \ref{bidfuns_detail}. The efficiency ($Eff.$) -- the probability that the winner has a highest valuation among the bidders -- decreases because the discounted bidder sometimes outbids the others even if some of them have a slightly higher valuation. However, the impact on efficiency is mild, because the bid curves for the discounted bidders are not far to the left of those for the other bidders, as shown in Figure\ref{bidfuns}. 

The win rate for the discounted bidder increases smoothly as the discount increases ($disc.$). The surplus ($E\;Surplus$) for the discounted bidder does the same, due to their price reduction, but partially offset by their increased equilibrium bids. The expected cost ($E\;Cost$) -- price after discount -- for the discounted bidder increases, but that is because of their increased win rate. To get expected cost given an auction win, divide the expected cost in the table by the win rate. For the discounted bidder, the expected cost given a win decreases as the discount rate increases. The opposite is true for the other bidders.

\section{Conclusion} \label{conc_sec} 
Comparing pre- and post-auction discounts, we showed that the same additive discount level produces the same equilibria. For multiplicative discounts, we showed that there are discounts that produce corresponding equilibria for pre- and post-discounting, but they are unequal: price reduction rate $r$ and bid augmentation rate $a$ have corresponding equilibria if $r = \frac{a}{1+a}$. So, in shifting between pre- and post-auction discounts, the rates need to be adjusted to maintain auction outcomes. We also explored how using the same pre- and post-auction discount rates alters auctions and how discounts can be set to compensate for asymmetry among bidders. 

Using empirical data from an online advertising auction, we studied how different discount rates would affect auction outcomes. To do that, we showed how to accurately estimate optimal bid functions under discounting for valuation distributions based on data. For that data, we found that discounting, even up to a 25\% price reduction, has only a mild negative impact (about 1.1\%) on revenue while shifting almost 5\% of the auction wins to the discounted bidder. 

In the future, it would be interesting to apply the methods we use in this paper to evaluate impacts on auction outcomes if discounted bidders have valuation distributions that are very different from the other bidders' distributions -- either much lower or higher on average or with different levels of variation. For those cases, it would be interesting to model bidders' decisions to enter or exit the auction and add that to the computation of estimated auction outcomes. Also, for cases with complicated valuation distributions, it could be useful to employ  higher-order methods rather than Euler's method to estimate bid functions, or to use a curve-fitting approach \cite{gentry18}.

\bibliographystyle{ACM-Reference-Format}
\bibliography{bax} 

\end{document}